\documentclass[submission,copyright,creativecommons]{eptcs}
\providecommand{\event}{PLACES 2023} 

\usepackage{etoolbox}
\usepackage{xspace}
\usepackage{amsmath,amssymb,amsthm}
\usepackage{stmaryrd}
\usepackage{xcolor}
\usepackage{cleveref}
\usepackage{iftex}
\usepackage{mathpartir}
\usepackage{upgreek}
\usepackage{mathtools}
\usepackage{prooftree}

\ifpdf
  \usepackage{underscore}         
  \usepackage[T1]{fontenc}        
\else
  \usepackage{breakurl}           
\fi

\title{A Logical Account of Subtyping for Session Types}
\author{Ross Horne
\institute{University of Luxembourg}
\and
Luca Padovani
\institute{University of Camerino}
}
\def\titlerunning{A Logical Account of Subtyping for Session Types}
\def\authorrunning{R. Horne and L. Padovani}

\newif\ifcomments
\commentstrue

\newcommand{\cf}{\textit{cf.}\xspace}

\newcommand{\mktag}[1]{\mathsf{\color{teal}#1}}
\newcommand{\mkkeyword}[1]{\mathsf{\color{blue}#1}}
\newcommand{\mkfunction}[1]{\mathsf{#1}}

\newcommand{\parens}[1]{(#1)}
\newcommand{\braces}[1]{\{#1\}}
\newcommand{\bracks}[1]{[#1]}
\newcommand{\angles}[1]{\langle#1\rangle}

\newcommand{\set}[1]{\braces{#1}}
\newcommand{\seqof}[1]{\overline{#1}}

\newcommand{\rulename}[1]{\textnormal{\textsc{\small[#1]}}}

\newcommand{\defrule}[2][]{\hypertarget{rule:\ifblank{#1}{#2}{#1}}{\rulename{#2}}}
\newcommand{\refrule}[2][]{\rulename{#2}\xspace}

\newcommand{\Calculus}{\muCPinf}
\newcommand{\muCPinf}{$\mu\textsf{CP}^\infty$\xspace}

\newcommand{\muMALL}{\texorpdfstring{$\mu\textsf{\upshape MALL}^\infty$}{muMALL}\xspace}

\newcommand{\CP}{$\textsf{CP}$\xspace}
\newcommand{\muCP}{$\mu\textsf{CP}$\xspace}

\newcommand\eoe{\hfill$\lrcorner$}

\newcommand{\existinf}{\exists^\infty}

\newcommand{\NatSet}{\mathbb{N}}

\newcommand\A{\mathsf{A}}
\newcommand\B{\mathsf{B}}

\newcommand{\InTag}{\mktag{in}}

\newcommand{\EndTag}{\mktag{end}}
\newcommand{\AddTag}{\mktag{add}}

\newcommand{\x}{x}
\newcommand{\y}{y}
\newcommand{\z}{z}
\renewcommand{\u}{u}
\renewcommand{\v}{v}

\newcommand{\parop}{\mathbin|}

\newcommand{\Let}[3]{#1\ifblank{#2}{}{\parens{#2}}\triangleq#3}
\newcommand{\Call}[2]{#1\ifblank{#2}{}{\angles{#2}}}

\newcommand{\Close}[1]{#1\bracks{}}
\newcommand{\Wait}[1]{#1\parens{}}
\newcommand{\Fail}[1]{\mkkeyword{fail}~#1}
\newcommand{\Select}[2]{#1\bracks{#2}}
\newcommand{\Left}[1]{\Select{#1}{\InTag_0}}
\newcommand{\Right}[1]{\Select{#1}{\InTag_1}}
\newcommand{\CaseX}[3]{\mkkeyword{case}~{#1}\braces{#2,#3}}
\newcommand{\Case}[3]{\CaseX{#1}{#2}{#3}}
\newcommand{\Send}[4]{#1\bracks{#2}\parens{#3\parop#4}}

\newcommand{\Receive}[2]{#1\parens{#2}}
\newcommand{\Cut}[3]{\parens{#1}\parens{#2\parop#3}}

\newcommand{\Type}{\TypeA}
\newcommand{\TypeA}{A}
\newcommand{\TypeB}{B}

\newcommand{\TypeSet}{\mathcal{T}}

\newcommand{\X}{X}
\newcommand{\Y}{Y}

\newcommand{\mkformula}[1]{#1}

\newcommand{\Bot}{\mkformula\bot}
\newcommand{\Top}{\mkformula\top}
\newcommand{\One}{\mkformula{\mathbf{1}}}
\newcommand{\Zero}{\mkformula{\mathbf{0}}}
\newcommand{\choice}{\mathbin{\mkformula\oplus}}
\newcommand{\branch}{\mathbin{\mkformula\binampersand}}
\newcommand{\xbranch}[1]{{\branch}\set{#1}}
\newcommand{\xchoice}[1]{{\choice}\set{#1}}
\newcommand{\tfork}{\mathbin{\mkformula\otimes}}
\newcommand{\tjoin}{\mathbin{\mkformula\bindnasrepma}}

\newcommand{\tmu}{\mkformula\mu}
\newcommand{\tnu}{\mkformula\nu}

\newcommand{\Num}{\mathsf{Num}}

\newcommand{\Context}{\ContextC}
\newcommand{\ContextC}{\Upgamma}
\newcommand{\ContextD}{\Updelta}

\newcommand{\wtp}[3][]{#2 \vdash\ifblank{#1}{}{_{#1}} #3}

\newcommand{\SubRefl}{refl}
\newcommand{\SubBot}{bot}
\newcommand{\SubTop}{top}
\newcommand{\SubCong}{cong}
\newcommand{\SubLeft}[1]{left-$#1$}
\newcommand{\SubRight}[1]{right-$#1$}

\newcommand{\CallRule}{call}

\newcommand{\FailRule}{$\Top$}
\newcommand{\CloseRule}{$\One$}
\newcommand{\WaitRule}{$\Bot$}
\newcommand{\SendRule}{$\tfork$}
\newcommand{\ReceiveRule}{$\tjoin$}
\newcommand{\SelectRule}{$\choice$}
\newcommand{\CaseRule}{$\branch$}
\newcommand{\FoldRule}[1]{$#1$}

\newcommand{\CoRecRule}{$\tnu$}
\newcommand{\SubRule}{sub}

\newcommand{\pcong}{\preccurlyeq}
\newcommand{\red}{\rightarrow}

\newcommand{\wred}{\Rightarrow}
\newcommand{\nred}{\arrownot\red}

\newcommand{\eqdef}{\stackrel{\smash{\textsf{\upshape\tiny def}}}=}

\newcommand{\subf}{\preceq}

\newcommand{\subt}{\leqslant}

\newcommand{\fn}[1]{\mkfunction{fn}\parens{#1}}

\newcommand{\subst}[2]{\braces{#1/#2}}
\newcommand{\dual}[1]{#1^{\bot}}
\newcommand{\dom}[1]{\mkfunction{dom}\parens{#1}}
\newcommand{\InfOften}[1]{\mkfunction{inf}\parens{#1}}

\newcommand{\sem}[1]{\left\llbracket#1\right\rrbracket}

\theoremstyle{plain}
\newtheorem{theorem}{Theorem}[section]

\newtheorem{corollary}{Corollary}[section]

\theoremstyle{definition}
\newtheorem{definition}{Definition}[section]
\newtheorem{example}{Example}[section]

\theoremstyle{remark}

\theoremstyle{remark}

\begin{document}
\maketitle

\begin{abstract}
    We study the notion of subtyping for session types in a logical setting,
    where session types are propositions of multiplicative/additive linear logic
    extended with least and greatest fixed points. The resulting subtyping
    relation admits a simple characterization that can be roughly spelled out as
    the following lapalissade: every session type is larger than the smallest
    session type and smaller than the largest session type. At the same time, we
    observe that this subtyping, unlike traditional ones, preserves termination
    in addition to the usual safety properties of sessions. We present a
    calculus of sessions that adopts this subtyping relation and we show that
    subtyping, while useful in practice, is superfluous in the theory: every use
    of subtyping can be ``compiled away'' via a coercion semantics.
\end{abstract}
    
\section{Introduction}
\label{sec:introduction}

Session types~\cite{Honda93,HondaVasconcelosKubo98,HuttelEtAl16} are
descriptions of communication protocols supported by an elegant correspondence
with linear logic~\cite{Wadler14,CairesPfenningToninho16,LindleyMorris16} that
provides session type systems with solid logical foundations.
As an example, below is the definition of a session type describing the protocol
implemented by a mathematical server (in the examples of this section, $\branch$
and $\choice$ are $n$-ary operators denoting external and internal labeled
choices, respectively):
\[
    B =
    \xbranch{
        \EndTag : \Bot,
        \AddTag : \dual\Num \tjoin \dual\Num \tjoin \Num \tfork B
    }
\]

According to the session type $B$, the server first waits for a label -- either
$\EndTag$ or $\AddTag$ -- that identifies the operation requested by the client.
If the label is $\EndTag$, the client has no more requests and the server
terminates. If the label is $\AddTag$, the server waits for two numbers, sends
their sum back to the client and then makes itself available again offering the
same protocol $B$. In this example, we write $\dual\Num$ for the type of numbers
being consumed and $\Num$ for the type of numbers being produced.
A client of this server could implement a communication protocol described by
the following session type:
\[
A =
\xchoice{
    \AddTag : \Num \tfork \Num \tfork \dual\Num \tjoin
              \xchoice{ \EndTag : \One }
}
\]

This client sends the label $\AddTag$ followed by two numbers, it
receives the result and then terminates the interaction with the server by
sending the label $\EndTag$.
When we connect two processes through a session, we expect their interaction to
be flawless. In many session type systems, this is guaranteed by making sure
that the session type describing the behavior of one process is the \emph{dual}
of the session type describing the behavior of its peer. \emph{Duality}, often
denoted by $\dual{\,\cdot\,}$, is the operator on session types that inverts the
\emph{direction} of messages without otherwise altering the structure of
protocol.
In the above example it is clear that $A$ is \emph{not} the dual of $B$ nor is
$B$ the dual of $A$. Nonetheless, we would like such client and such server to
be declared compatible, since the client is exercising only a subset of the
capabilities of the server.
To express this compatibility we have to resort to a more complex relation
between $A$ and $B$, either by observing that $B$ (the behavior of the server)
is a \emph{more accommodating} version of $\dual{A}$ or by observing that $A$
(the behavior of the client) is a \emph{less demanding} version of $\dual{B}$.
We make these relations precise by means of a \emph{subtyping relation} $\subt$
for session types. Subtyping enhances the applicability of type systems by means
of the well-known substitution principle: an entity of type $C$ can be used
where an entity of type $D$ is expected if $C$ is a subtype of $D$. After the
initial work of Gay and Hole~\cite{GayHole05} many subtyping relations for
session types have been
studied~\cite{CastagnaDezaniGiachinoPadovani09,Padovani13,MostrousYoshida15,Padovani16,GhilezanEtAl21}.
Such subtyping relations differ widely in the way they are defined and/or in the
properties they preserve, but they all share the fact that subtyping is
essentially defined by the branching structure of session types given by labels.
To illustrate this aspect, let us consider again the session types $A$ and $B$
defined above. We have
\begin{equation}
    \label{eq:ab}
    B \subt
    \xbranch{
        \AddTag : \dual\Num \tjoin \dual\Num \tjoin \Num \tfork
        \xbranch{
            \EndTag : \Bot
        }
    }
    = \dual{A}
\end{equation}
meaning that a server behaving as $B$ can be safely used where a server behaving
as $\dual{A}$ is expected. Dually, we also have
\begin{equation}
    \label{eq:ba}
    {A} \subt
    \xchoice{
        \EndTag : \One,
        \AddTag : \Num \tfork \Num \tfork \dual\Num \tjoin \dual{B}
    }
    = \dual{B}
\end{equation}
meaning that a client behaving as $A$ can be safely used where a client behaving
as $\dual{B}$ is expected.
Note how subtyping is crucially determined by the sets of labels that can be
received/sent when comparing two related types.
In \eqref{eq:ab}, the server of type $B$ is willing to accept any label from the
set $\set{\EndTag,\AddTag}$, which is a \emph{superset} of $\set\AddTag$ that we
have in $\dual{A}$.
In \eqref{eq:ba}, the client is (initially) sending a label from the set
$\set\AddTag$, which is a subset of $\set{\EndTag,\AddTag}$ that we have in
$\dual{B}$.
This co/contra variance of labels in session types is a key distinguishing
feature of all known notions of subtyping for session types.\footnote{Gay and
Hole~\cite{GayHole05} and other
authors~\cite{CastagnaDezaniGiachinoPadovani09,Padovani13,Padovani16} define
subtyping for session types in such a way that the \emph{opposite} relations of
\cref{eq:ab,eq:ba} hold. Both viewpoints are viable depending on whether session
types are considered to be types of \emph{channels} or types of
\emph{processes}. Here we take the latter stance, referring to Gay~\cite{Gay16}
for a comparison of the two approaches.}

In this work we study the notion of subtyping for session types in a setting
where session types are propositions of
\muMALL~\cite{BaeldeDoumaneSaurin16,Doumane17}, the infinitary proof theory of
multiplicative additive linear logic extended with least and greatest fixed
points.
Our investigation has two objectives.
First, to understand whether and how it is possible to capture the well-known
co/contra variance of behaviors when the connectives used to describe branching
session types ($\branch$ and $\choice$ of linear logic) have fixed arity.
Second, to understand whether there are criticial aspects of subtyping that
become relevant when typing derivations are meant to be logically sound.


At the core of our proposal is the observation that, when session types (hence
process behaviors) are represented by linear logic
propositions~\cite{Wadler14,CairesPfenningToninho16,LindleyMorris16}, it is
impossible to write a process that behaves as $\Zero$ and it is very easy to
write a process that behaves as $\Top$. If we think of a session type as the set
of processes that behave according to that type, this means that the additive
constants $\Zero$ and $\Top$ may serve well as the least and greatest elements
of a session subtyping relation. Somewhat surprisingly, the subtyping relation
defined by these properties of $\Zero$ and $\Top$ allows us to express
essentially the same subtyping relations that arise from the usual co/contra
variance of labels.
For example, following our proposal the session type of the client, previously
denoted $A$, would instead be written as
\[
    C = \xchoice{
        \EndTag : \Zero,
        \AddTag : \Num \tfork \Num \tfork \dual\Num \tjoin
        \xchoice{
            \EndTag : \One,
            \AddTag : \Zero
        }
    }
\]
using which we can derive both
\[
    B \subt
    \xbranch{
        \EndTag : \Top,
        \AddTag : \dual\Num \tjoin \dual\Num \tjoin \Num \tfork
        \xbranch{
            \EndTag : \Bot,
            \AddTag : \Top
        }
    }
    = \dual{C}
    \text{\qquad as well as\qquad}
    C \subt \dual{B}
\]
without comparing labels and just using the fact that $\Zero$ is the least
session type and $\Top$ the greatest one.
Basically, instead of \emph{omitting those labels} that correspond to impossible
continuations (\cf the missing $\EndTag$ and $\AddTag$ in $A$), we use the
uninhabited session type $\Zero$ or its dual $\Top$ \emph{as impossible
continuations} (\cf~$C$). It could be argued that the difference between the two
approaches is mostly cosmetic. Indeed, it is easy to devise (de)sugaring
functions to rewrite session types from one syntax to the other.
However, the novel approach we propose allows us to recast the well-known
subtyping relation for session types in a logical setting. A first consequence
of this achievement is that the soundness of the type system \emph{with
subtyping} does not require an \emph{ad hoc} proof, but follows from the
soundness of the type system \emph{without subtyping} through a suitable
coercion semantics. In addition, we find out that the subtyping relation we
propose preserves not only the usual \emph{safety properties} -- communication
safety, protocol fidelity and deadlock freedom -- but also \emph{termination},
which is a \emph{liveness property}.


\paragraph{Structure of the paper.}
In \Cref{sec:language} we define \Calculus, a session calculus of processes
closely related to \muCP~\cite{LindleyMorris16} and \CP~\cite{Wadler14}.
In \Cref{sec:types} we define the type language for \Calculus and the subtyping
relation.
In \Cref{sec:typing-rules} we define the typing rules for \Calculus and give a
coercion semantics to subtyping, thus showing that the type system of \Calculus
is a conservative extension of \muMALL~\cite{BaeldeDoumaneSaurin16,Doumane17}.
We wrap up in \Cref{sec:conclusion}.
%

\newcommand\Buyer{\textit{Buyer}}
\newcommand\Seller{\textit{Seller}}

\section{Syntax and semantics of \texorpdfstring\Calculus{calculus}}
\label{sec:language}

\begin{table}
    \centering
    \begin{math}
        \displaystyle
        \begin{array}[t]{@{}r@{~}c@{~}ll@{}}
            \textbf{Process} \quad
            P, Q & ::= & &
            \\
            &   & \Call\A{\seqof\x} & \text{invocation}
            \\
            & | & \Wait\x.P & \text{signal input}
            \\
            & | & \Receive\x\y.P & \text{channel input}
            \\
            & | & \Case\x{P}{Q} & \text{choice input}
        \end{array}
        ~
        \begin{array}[t]{@{}r@{~}c@{~}lll@{}}
            & | & \Cut\x{P}{Q} & \text{composition}
            \\
            & | & \Fail\x & \text{failure}
            \\
            & | & \Close\x & \text{signal output}
            \\
            & | & \Send\x\y{P}{Q} & \text{channel output}
            \\
            & | & \Select\x{\InTag_i}.P & \text{choice output} & i\in\set{0,1}
        \end{array}
    \end{math}
    \caption{\label{tab:syntax} Syntax of \Calculus.}
\end{table}

The syntax of \Calculus is shown in \Cref{tab:syntax} and makes use of a set of
\emph{process names} $\A$, $\B$, \ldots and of an infinite set of \emph{channels}
$x$, $y$, $z$ and so on.
The calculus includes standard forms representing communication actions:
$\Fail\x$ models a process failing on $x$; $\Wait\x.P$ and $\Close\x$ model the
input/output of a termination signal on $x$; $\Case\x{P}{Q}$ and
$\Select\x{\InTag_i}.P$ model the input/output of a label $\InTag_i$ on $x$;
$\Receive\x\y.P$ and $\Send\x\y{P}{Q}$ model the input/output of a channel $y$
on $x$. Note that $\Send\x\y{P}{Q}$ outputs a \emph{new} channel $y$ which is
bound in $P$ but not in $Q$. Free channel output can be encoded as shown in
previous works~\cite{LindleyMorris16}.
%
%
The form $\Cut\x{P}{Q}$ models a session $x$ connecting two parallel processes
$P$ and $Q$ and the form $\Call\A{\seqof\x}$ models the invocation of the
process named $\A$ with parameters $\seqof\x$. For each process name $\A$ we
assume that there is a unique global definition of the form
$\Let\A{\seqof\x}{P}$ that gives its meaning. Hereafter $\seqof\x$ denotes a
possibly empty sequence of channels.
%
%
The notions of free and bound channels are defined in the expected way. We
identify processes up to renaming of bound channels and we write $\fn{P}$ for
the set of free channels of $P$.

\begin{table}
  \[
    \begin{array}{@{}rr@{~}c@{~}ll@{}}
      \defrule{s-par-comm} &
      \Cut\x{P}{Q} & \pcong & \Cut\x{Q}{P}
      \\
      \defrule{s-par-assoc} &
      \Cut\x{P}{\Cut\y{Q}{R}} & \pcong & \Cut\y{\Cut\x{P}{Q}}{R} &
      x\in\fn{Q}\setminus\fn{R}, y\not\in\fn{P}
      \\
      \defrule{s-call} &
      \Call\A{\seqof\x} & \pcong & P & \Let\A{\seqof\x}P
      \\
      \\
      \defrule{r-close} &
      \Cut\x{\Close\x}{\Wait\x.P} & \red & P
      \\
      \defrule{r-comm} &
      \Cut\x{\Send\x\y{P}{Q}}{\Receive\x\y.R} & \red & \Cut\y{P}{\Cut\x{Q}{R}}
      \\
      \defrule{r-case} &
      \Cut\x{\Select\x{\InTag_i}.P}{\Case\x{Q_0}{Q_1}} & \red & \Cut\x{P}{Q_i}
      & i \in \set{0,1}
      \\
      \defrule{r-par} &
      \Cut\x{P}{R} & \red & \Cut\x{Q}{R} &
      P \red Q
      \\
      \defrule{r-struct} &
      P & \red & Q &
      P \pcong P' \red Q' \pcong Q
    \end{array}
  \]
  \caption{\label{tab:semantics}Structrual pre-congruence and reduction semantics of \Calculus.}
\end{table}

The operational semantics of \Calculus is shown in \Cref{tab:semantics} and
consists of a structural pre-congruence relation $\pcong$ and a reduction
relation $\red$, both of which are fairly standard. We write $P \red$ if $P \red
Q$ for some $Q$ and we say that $P$ is \emph{stuck}, notation $P \nred$, if not
$P \red$.

\newcommand{\Client}{\mathsf{Client}}
\newcommand{\Server}{\mathsf{Server}}

\begin{example}
  \label{ex:math}
  We can model client and server described in \Cref{sec:introduction} as the
  processes below.
  \[
    \Let\Client{x}{
      \Right\x.
      \Left\x.
      \Close\x
    }
    \qquad
    \Let\Server{x,z}{
      \Case\x{
        \Wait\x.\Close\z
      }{
        \Call\Server{x,z}
      }
    }
  \]
  For simplicity, we only focus on the overall structure of the processes rather
  than on the actual mathematical operations they perform, so we omit any
  exchange of concrete data from this model.
  \eoe
  %
\end{example}

We conclude this section with the definitions of the properties ensured by our
type system, namely \emph{deadlock freedom} and \emph{termination}. The latter
notion is particularly relevant in our setting since termination preservation is
a novel aspect of the subtyping relation that we are about to define.

\begin{definition}[deadlock-free process]
  \label{def:df}
  We say that $P$ is \emph{deadlock free} if $P \wred Q \nred$ implies that $Q$
  is not (structurally pre-congruent to) a process of the form
  $\Cut\x{R_1}{R_2}$.
\end{definition}

A deadlock-free process either reduces or it is stuck waiting to synchronize on
some free channel.

\begin{definition}[terminating process]
  \label{def:termination}
  A \emph{run} of a process $P$ is a (finite or infinite) sequence
  $(P_0,P_1,\dots)$ of processes such that $P_0 = P$ and $P_i \red P_{i+1}$
  whenever $i+1$ is a valid index of the sequence. We say that a run is maximal
  if either it is infinite or if the last process in it is stuck. We say that
  $P$ is \emph{terminating} if every maximal run of $P$ is finite.
\end{definition}

Note that a terminating process is not necessarily free of restrictions. For
example, $\Cut\x{\Fail\x}{\Close\x}$ is terminated but not deadlock free. It
really is the conjunction of deadlock freedom and termination (as defined above)
that ensure that a process is ``well behaved''.


\section{Types and subtyping}
\label{sec:types}

The type language for \Calculus consists of the propositions of
\muMALL~\cite{BaeldeDoumaneSaurin16,Doumane17,BaeldeEtAl22}, the infinitary
proof theory of multiplicative/additive linear logic extended with least and
greatest fixed points. We start from the definition of \emph{pre-types}, which
are linear logic propositions built using type variables taken from an infinite
set and ranged over by $X$ and $Y$.
\[
    \textbf{Pre-type}
    \qquad
    A, \TypeB ::= X \mid \Bot \mid \One \mid \Top \mid \Zero \mid A \tjoin \TypeB \mid A \tfork \TypeB \mid A \branch \TypeB \mid A \choice \TypeB \mid \tnu\X.\Type \mid \tmu\X.\Type
\]

The usual notions of free and bound type variables apply. A \emph{type} is a
closed pre-type. We assume that type variables occurring in types are
\emph{guarded}. That is, we forbid types of the form
$\sigma_1\X_1\dots\sigma_n\X_n.X_i$ where $\sigma_1,\dots,\sigma_n \in
\set{\mu,\nu}$.
We write $\dual\Type$ for the \emph{dual} of $\Type$, which is defined in the
expected way with the proviso that $\dual\X = X$. This way of dualizing type
variables is not problematic since we will always apply $\dual{\,\cdot\,}$ to
types, which contain no free type variables. As usual, we write
$A\subst\TypeB\X$ for the (pre-)type obtained by replacing every $X$ occurring
free in the pre-type $A$ with the type $\TypeB$.
Hereafter we let $\kappa$ range over the constants $\Zero$, $\One$, $\Bot$ and
$\Top$, we let $\star$ range over the connectives $\branch$, $\choice$, $\tjoin$
and $\tfork$ and $\sigma$ range over the binders $\mu$ and $\nu$.
Also, we say that any type of the form $\sigma\X.\Type$ is a $\sigma$-type.

We write $\subf$ for the standard \emph{sub-formula} relation on types.
To be precise, the relation $\subf$ is the least preorder on types such that $A
\subf \sigma\X.A$ and $A_i \subf A_1 \star A_2$.
For example, consider $A \eqdef \mu\X.\nu\Y.(1 \choice X)$ and its unfolding $A'
\eqdef \nu\Y.(1 \choice A)$. We have $A \subf 1 \choice A \subf A'$, hence $A$
is a sub-formula of $A'$.
Given a set $\TypeSet$ of types we write $\min\TypeSet$ for the $\subf$-minimum
type in $\TypeSet$ when it is defined.

\begin{table}
  \begin{mathpar}
    \inferrule[\defrule\SubBot]{~}{
      \Zero \subt A
    }
    \and
    \inferrule[\defrule\SubTop]{~}{
      A \subt \Top
    }
    \and
    \inferrule[\defrule\SubRefl]{~}{
      \kappa \subt \kappa
    }
    \and
    \inferrule[\defrule\SubCong]{
      A \subt A'
      \\
      B \subt B'
    }{
      A \star B \subt A' \star B'
    }
    \and
    \inferrule[\defrule{\SubLeft\sigma}]{
      A\subst{\sigma\X.A}\X \subt B
    }{
      \sigma\X.A \subt B
    }
    \and
    \inferrule[\defrule{\SubRight\sigma}]{
      A \subt B\subst{\sigma\X.B}\X
    }{
      A \subt \sigma\X.B
    }
  \end{mathpar}
  \caption{\label{tab:subt}Subtyping for session types.}
\end{table}

\Cref{tab:subt} shows the inference rules for subtyping judgments. The rules are
meant to be interpreted coinductively so that a judgment $A \subt B$ is
derivable if it is the conclusion of a finite/infinite derivation.
The rules \refrule\SubBot and \refrule\SubTop establish that $\Zero$ and $\Top$
are respectively the least and the greatest session type; the rules
\refrule\SubRefl and \refrule\SubCong establish reflexivity and pre-congruence
of $\subt$ with respect to all the constants and connectives; the rules
\refrule{\SubLeft\sigma} and \refrule{\SubRight\sigma} allow fixed points to be
unfolded on either side of $\subt$.

\begin{example}
  \label{ex:subt}
  Consider the types $A \eqdef \Zero \choice (\One \choice \Zero)$ and $\TypeB
  \eqdef \nu\X.(\Bot \branch X)$ which, as we will see later, describe the
  behavior of $\Client$ and $\Server$ in \Cref{ex:math}. We can derive both
  $\TypeA \subt \dual\TypeB$ and $\TypeB \subt \dual\TypeA$ thus:
  \[
    \begin{prooftree}
      \[
        \[
          \justifies
          \Zero \subt \One
          \using\refrule\SubBot
        \]
        \[
          \[
            \[
              \justifies
              \One \subt \One
              \using\refrule\SubRefl
            \]
            \[
              \justifies
              \Zero \subt \dual\TypeB
              \using\refrule\SubBot
            \]
            \justifies
            \One \choice \Zero \subt \One \choice \dual\TypeB
            \using\refrule\SubCong
          \]
          \justifies
          \One \choice \Zero \subt \dual\TypeB
          \using\refrule{\SubRight\mu}
        \]
        \justifies
        \TypeA \subt \One \choice \dual\TypeB
        \using\refrule\SubCong
      \]
      \justifies
      \TypeA \subt \dual\TypeB
      \using\refrule{\SubRight\mu}
    \end{prooftree}
    \qquad
    \begin{prooftree}
      \[
        \[
          \justifies
          \Bot \subt \Top
          \using\refrule\SubTop
        \]
        \[
          \[
            \[
              \justifies
              \Bot \subt \Bot
              \using\refrule\SubRefl
            \]
            \[
              \justifies
              \TypeB \subt \Top
              \using\refrule\SubTop
            \]
            \justifies
            \Bot \tjoin \TypeB \subt \Bot \tjoin \Top
            \using\refrule\SubCong
          \]
          \justifies
          \TypeB \subt \Bot \branch \Top
          \using\refrule{\SubLeft\nu}
        \]
        \justifies
        \Bot \tjoin \TypeB \subt \dual\TypeA
        \using\refrule\SubCong
      \]
      \justifies
      \TypeB \subt \dual\TypeA
      \using\refrule{\SubLeft\nu}
    \end{prooftree}
  \]
  %
\end{example}

The rules \refrule{\SubLeft\sigma} and \refrule{\SubRight\sigma} may look
suspicious since they are applicable to either side of $\subt$ regardless of the
intuitive interpretation of $\mu$ and $\nu$ as least and greatest fixed points.
In fact, if subtyping were solely defined by the derivability according to the
rules in \Cref{tab:subt}, the two fixed point operators would be equivalent. For
example, both $\mu\X.(\One \choice X) \subt \nu\X.(\One \choice X)$ and
$\nu\X.(\One \choice X) \subt \mu\X.(\One \choice X)$ are derivable even though
only the first relation seems reasonable. We will see in \Cref{ex:chatter} that
allowing the second relation is actually \emph{unsound}, in the sense that it
compromises the termination property enjoyed by well-typed processes.
We obtain a sound subtyping relation by ruling out some infinite derivations as
per the following (and final) definition of subtyping.

\begin{definition}[subtyping]
  \label{def:subt}
  We say that $A$ is a \emph{subtype} of $B$ if $A \subt B$ is derivable and,
  for every infinite branch $(A_i \subt B_i)_{i\in\NatSet}$ of the derivation,
  either (1) $\min\set{C \mid \existinf i: A_i = C}$ is a $\mu$-type or (2)
  $\min\set{C \mid \existinf i: B_i = C}$ is a $\nu$-type. Hereafter $\existinf
  i$ means the existence of infinitely many $i$'s with the stated property.
\end{definition}

The clauses (1) and (2) of \Cref{def:subt} make sure that $\mu$ and $\nu$ are
correctly interpreted as least and greatest fixed points. In particular, we
expect the least fixed point to be subsumed by a greatest fixed point, but not
vice versa in general. For example, consider once again the (straightforward)
derivations for the aforementioned subtyping judgments $\mu\X.(\One \choice X)
\subt \nu\X.(\One \choice X)$ and $\nu\X.(\One \choice X) \subt \mu\X.(\One
\choice X)$. The first derivation satisfies both clauses (there is only one
infinite branch, along which a $\mu$-type is unfolded infinitely many times on
the left hand side of $\subt$ and a $\nu$-type is unfolded infinitely many times
on the right hand side of $\subt$). The second derivation satisfies neither
clause. Therefore, $\mu\X.(\One \choice X)$ is a subtype of $\nu\X.(\One \choice
X)$ but $\nu\X.(\One \choice X)$ is not a subtype of $\mu\X.(\One \choice X)$.
As we will see in \Cref{sec:typing-rules}, the application of a subtyping
relation $A \subt B$ can be explicitly modeled as a process \emph{consuming} a
channel of type $A$ while \emph{producing} a channel of type $B$. According to
this interpretation of subtyping, we can see that clause (1) of \Cref{def:subt}
is just a dualized version of clause~(2).

In both clauses of \Cref{def:subt} there is a requirement that the type of the
fixed point on each side of the relation is determined by the $\subf$-minimum of
the types that appear infinitely often on either side. This is needed to handle
correctly alternating fixed points, by determining which one is actively
contributing to the infinite path. To see what effect this has consider the
types $A \eqdef \mu\X.\nu\Y.(\One \choice \X)$, $A' \eqdef \nu\Y.(\One \choice
A)$, $B \eqdef \mu\X.\mu\Y.(\One \choice \X)$ and $B' \eqdef \mu\Y.(\One \choice
B)$.
Observe that $A$ unfolds to $A'$, $A'$ unfolds to $\One \choice A$, $B$ unfolds
to $B'$ and $B'$ unfolds to $\One \choice B$. We have $A \subt B$ despite $Y$ is
bound by a \emph{greatest} fixed point on the left and by a \emph{least} fixed
point on the right. Indeed, both $A$ and $A'$ occur infinitely often in the
(only) infinite branch of the derivation for $A \subt B$, but $A \subf A'$
according to the intuition that the $\subf$-minimum type that occurs infinitely
often is the one corresponding to the outermost fixed point. In this case, the
outermost fixed point is $\mu\X$ which ``overrides'' the contribution of the
inner fixed point $\nu\Y$. The interested reader may refer to the literature on
\muMALL~\cite{BaeldeDoumaneSaurin16,Doumane17} for details.

Hereafter, unless otherwise specified, we write $A \subt B$ to imply that $A$ is
a subtype of $B$ and not simply that the judgment $A \subt B$ is derivable.
It is possible to show that $\subt$ is a preorder and that $A \subt \TypeB$
implies $\dual\TypeB \subt \dual\TypeA$.
Indeed, as illustrated in \Cref{ex:subt}, we obtain a derivation of $\dual{B}
\subt \dual{A}$ from that of $A \subt B$ by dualizing every judgment and by
turning every application of \refrule{\SubLeft\sigma} (respectively
\refrule{\SubRight\sigma}, \refrule\SubBot, \refrule\SubTop) into an application
of \refrule{\SubRight{\dual\sigma}} (respectively
\refrule{\SubLeft{\dual\sigma}}, \refrule\SubTop, \refrule\SubBot).

\section{Typing rules}
\label{sec:typing-rules}

\begin{table}
    \begin{mathpar}
        \inferrule[\defrule\CallRule]{
            \wtp{P}{\seqof{x : \Type}}
        }{
            \wtp{\Call\A{\seqof\x}}{\seqof{x : \Type}}
        }
        ~\Let\A{\seqof\x}{P}
        \and
        \inferrule[\defrule\SubRule]{
          \wtp{P}{\ContextC, x : A}
          \\
          \wtp{Q}{\ContextD, x : B}
        }{
          \wtp{\Cut\x{P}{Q}}{\ContextC,\ContextD}
        }
        ~A \subt \dual{B}
        \and
        \inferrule[\defrule\FailRule]{~}{
            \wtp{\Fail\x}{\Context, x : \Top}
        }
        \and
        \inferrule[\defrule\WaitRule]{
            \wtp{P}{\Context}
        }{
            \wtp{\Wait\x.P}{\Context, x : \Bot}
        }
        \and
        \inferrule[\defrule\CloseRule]{~}{
            \wtp{\Close\x}{x : \One}
        }
        \and
        \inferrule[\defrule\ReceiveRule]{
            \wtp{P}{\Context, y : \TypeA, x : \TypeB}
        }{
            \wtp{\Receive\x\y.P}{\Context, x : \TypeA \tjoin \TypeB}
        }
        \and
        \inferrule[\defrule\SendRule]{
            \wtp{P}{\ContextC, y : \TypeA}
            \\
            \wtp{Q}{\ContextD, x : \TypeB}
        }{
            \wtp{\Send\x\y{P}{Q}}{\ContextC, \ContextD, x : \TypeA \tfork \TypeB}
        }
        \and
        \inferrule[\defrule\CaseRule]{
            \wtp{P}{\Context, x : \TypeA}
            \\
            \wtp{Q}{\Context, x : \TypeB}
        }{
            \wtp{\Case\x{P}{Q}}{\Context, x : \TypeA \branch \TypeB}
        }
        \and
        \inferrule[\defrule\SelectRule]{
          \wtp{P}{\Context, x : \Type_i}
        }{
          \wtp{\Select\x{\InTag_i}.P}{\Context, x : \Type_0 \choice \Type_1}
        }
        ~i\in\set{0,1}
        \and
        \inferrule[\defrule{\FoldRule\sigma}]{
          \wtp{P}{\Context, x : A\subst{\sigma\X.A}\X}
        }{
          \wtp{P}{\Context, x : \sigma\X.A}
        }
        \and
    \end{mathpar}
    \caption{\label{tab:typing-rules} Typing rules for \Calculus.}
\end{table}

In this section we describe the typing rules for \Calculus.
Typing judgments have the form $\wtp{P}\Context$ where $P$ is a process and
$\Context$ is a typing context, namely a finite map from channels to types. We
can read this judgment as the fact that $P$ behaves as described by the types in
the range of $\Context$ with respect to the channels in the domain of
$\Context$.
We write $\dom\Context$ for the domain of $\Context$, we write $x : A$ for the
typing context with domain $\set\x$ that maps $x$ to $A$, we write
$\ContextC,\ContextD$ for the union of $\ContextC$ and $\ContextD$ when
$\dom\ContextC \cap \dom\ContextD = \emptyset$.
The typing rules of \Calculus are shown in \Cref{tab:typing-rules} and, with the
exception of \refrule\CallRule and \refrule\SubRule, they correspond to the
proof rules of \muMALL~\cite{BaeldeDoumaneSaurin16,Doumane17} in which the
context is the sequent being proved and the process is (almost) a syntactic
representation of the proof.
The rules for the multiplicative/additive constants and for the connectives are
standard. The rule~\refrule{\FoldRule\sigma} where $\sigma\in\set{\mu,\nu}$
simply unfolds fixed points regardless of their nature. The rule
\refrule\CallRule unfolds a process invocation into its definition, checking
that the invocation and the definition are well typed in the same context.
Finally, \refrule\SubRule checks that the composition $\Cut\x{P}{Q}$ is well
typed provided that $A$ (the behavior of $P$ with respect to $x$) is a subtype
of $\dual{B}$ (where $B$ is the behavior of $Q$ with respect to $x$).
In this sense \refrule\SubRule embeds the substitution principle induced by
$\subt$ since it allows a process behaving as $A$ to be used where a process
behaving as $\dual{B}$ is expected.
Note that the standard cut rule of \muMALL is a special case of \refrule\SubRule
because of the reflexivity of $\subt$.

Like in \muMALL, the rules are meant to be interpreted coinductively so that a
judgment $\wtp{P}\Context$ is deemed derivable if there is an arbitrary (finite
or infinite) derivation whose conclusion is $\wtp{P}\Context$.

\begin{example}
  \label{ex:math-typing}
  Let us show the typing derivations for the processes discussed in
  \Cref{ex:math}. To this aim, let $\TypeA \eqdef \Zero \choice (\One \choice
  \Zero)$ and $\TypeB \eqdef \tnu\X.(\Bot \branch X)$ and recall from
  \Cref{ex:subt} that $A \subt \dual{B}$. We derive:
  \[
    \begin{prooftree}
      \[
        \[
          \[
            \[
              \justifies
              \wtp{
                \Close\x
              }{
                \One
              }
              \using\refrule\CloseRule
            \]
            \justifies
            \wtp{
              \Left\x.
              \Close\x
            }{
              \One \choice \Zero
            }
            \using\refrule\SelectRule
          \]
          \justifies
          \wtp{
            \Right\x.
            \Left\x.
            \Close\x
          }{
            x : \TypeA
          }
          \using\refrule\SelectRule
        \]
        \justifies
        \wtp{
          \Call\Client\x
        }{
          x : \TypeA
        }
        \using\refrule\CallRule
      \]
      \[
        \[
          \[
            \[
              \[
                \justifies
                \wtp{
                  \Close\z
                }{
                  z : \One
                }
                \using\refrule\CloseRule
              \]
              \justifies
              \wtp{
                \Wait\x.\Close\z
              }{
                x : \Bot,
                z : \One
              }
              \using\refrule\WaitRule
            \]
            \[
              \vdots
              \justifies
              \wtp{
                \Call\Server{x,z}
              }{
                x : \TypeB,
                z : \One
              }
            \]
            \justifies
            \wtp{
              \Case\x{
                \Wait\x.\Close\z
              }{
                \Call\Server{x,z}
              }
            }{
              x : \Bot \branch \TypeB,
              z : \One
            }
            \using\refrule\CaseRule
          \]
          \justifies
          \wtp{
            \Case\x{
              \Wait\x.\Close\z
            }{
              \Call\Server{x,z}
            }
          }{
            x : \TypeB,
            z : \One
          }
          \using\refrule\CoRecRule
        \]
        \justifies
        \wtp{
          \Call\Server{x,z}
        }{
          x : \TypeB,
          z : \One
        }
        \using\refrule\CallRule
      \]      
      \justifies
      \wtp{
        \Cut\x{\Call\Client\x}{\Call\Server{x,z}}
      }{
        z : \One
      }
      \using\refrule\SubRule
    \end{prooftree}
  \]

  We can obtain a similar typing derivation by swapping $\Client$ and $\Server$
  and using the relation $B \subt \dual{A}$.
  Note that $\Client$ and $\Server$ cannot be composed directly using a standard
  cut since $\TypeA \ne \dual\TypeB$. So, the use of subtyping in the above
  typing derivation is important to obtain a well-typed composition.
  %
  \eoe
\end{example}

It is a known fact that not every \muMALL derivation is a valid
one~\cite{BaeldeDoumaneSaurin16,Doumane17,BaeldeEtAl22}. In order to
characterize the valid derivations we need some auxiliary notions which we
recall below.



\begin{definition}[thread]
  \label{def:thread}
  Let $\gamma = (\wtp{P_i}{\Context_i})_{i\in\NatSet}$ be an infinite branch in
  a typing derivation and recall that $\wtp{P_{i+1}}{\Context_{i+1}}$ is a
  premise of $\wtp{P_i}{\Context_i}$.
  A \emph{thread} of $\gamma$ is a sequence $(x_i)_{i\geq k}$ of channels such
  that $x_i \in \dom{\Context_i}$ and either $x_i = x_{i+1}$ or $P_i =
  \Send{x_i}{x_{i+1}}{P_{i+1}}{Q}$ or $P_i = \Receive{x_i}{x_{i+1}}.P_{i+1}$ for
  every $i\ge k$.
\end{definition}

Intuitively, a thread is an infinite sequence of channel names $(x_i)_{i\geq k}$
that are found starting from some position $k$ in an infinite branch
$(\wtp{P_i}{\Context_i})_{i\in\NatSet}$ and that pertain to the same session.
For example, consider the derivation in \Cref{ex:math-typing} and observe that
there is only one infinite branch, the rightmost one. The sequence
$(x,x,x,\dots)$ is a thread that starts right above the conclusion of the
derivation.

\begin{definition}[$\nu$-thread]
  \label{def:nu-thread}
  Given a branch $\gamma = (\wtp{P_i}{\Context_i})_{i\in\NatSet}$ and a thread
  $t = (x_i)_{i\geq k}$ of $\gamma$, we write $\InfOften{\gamma,t} \eqdef \set{
  \Type \mid \existinf i\geq k : \Context_i(x_i) = \Type }$. We say that $t$ is
  a \emph{$\nu$-thread} of $\gamma$ if $\min\InfOften{\gamma,t}$ is a
  $\nu$-type.
\end{definition}

Given a branch $\gamma = (\wtp{P_i}{\Context_i})_{i\in\NatSet}$ and a thread $t
= (x_i)_{i\geq k}$ of $\gamma$, the thread identifies an infinite sequence
$(\Context_i(x_i))_{i\geq k}$ of types. The set $\InfOften{\gamma,t}$ is the set
of those types that occur infinitely often in this sequence and
$\min\InfOften{\gamma,t}$ is the $\subf$-minimum among these types (it can be
shown that the minimum of any set $\InfOften{\gamma,t}$ is always
defined~\cite{Doumane17}). We say that $t$ is a $\nu$-thread if such minimum
type is a $\nu$-type.
In \Cref{ex:math-typing}, the thread $t = (x,x,x,\dots)$ identifies the sequence
$(\TypeB, \TypeB, \Bot \branch \TypeB,\TypeB,\dots)$ of types in which both
$\TypeB$ and $\Bot \branch \TypeB$ occur infinitely often. Since $\TypeB \subf
\Bot \branch \TypeB$ and $\TypeB$ is a $\nu$-type we conclude that $t$ is a
$\nu$-thread.

\begin{definition}[valid branch]
  \label{def:valid-branch}
  Let $\gamma = (\wtp{P_i}{\Context_i})_{i\in\NatSet}$ be an infinite branch of
  a typing derivation. We say that $\gamma$ is \emph{valid} if there is a
  $\nu$-thread $(x_i)_{i\geq k}$ of $\gamma$ such that \refrule{\FoldRule\nu} is
  applied to infinitely many of the $x_i$.
\end{definition}

\Cref{def:valid-branch} establishes that a branch is valid if it contains a
$\nu$-thread in which the $\nu$-type occurring infinitely often is also unfolded
infinitely often. This happens in \Cref{ex:math-typing}, in which the
\refrule\CoRecRule rule is applied infinitely often to unfold the type of $x$.
The reader familiar with the \muMALL literature may have spotted a subtle
difference between our notion of valid branch and the standard
one~\cite{BaeldeDoumaneSaurin16,Doumane17}. In \muMALL, a branch is valid only
provided that the $\nu$-thread in it is not ``eventually constant'', namely if
the greatest fixed point that defines the $\nu$-thread is unfolded infinitely
many times. This condition is satisfied by our notion of valid branch because of
the requirement that there must be infinitely many applications of
\refrule{\FoldRule\nu} concerning the names in the $\nu$-thread.
Now we can define the notion of valid typing derivation.

\begin{definition}[valid derivation]
  \label{def:valid-derivation}
  A typing derivation is \emph{valid} if so is every infinite branch in it.
\end{definition}

\newcommand{\Coercion}[3]{\sem{#1}_{#2,#3}}

\begin{table}
  \begin{mathpar}
    \Coercion{
      \inferrule{~}{
        \Zero \subt A
      }
    }\x\y \triangleq \Fail\x
    \and
    \Coercion{
      \inferrule{~}{
        \One \subt \One
      }
    }\x\y \triangleq \Wait\x.\Close\y
    \and
    \Coercion{
      \inferrule{
        \pi_1 :: \TypeA \subt \TypeA'
        \\
        \pi_2 :: \TypeB \subt \TypeB'
      }{
        \TypeA \choice \TypeB \subt \TypeA' \choice \TypeB'
      }
    }\x\y \triangleq \Case\x{\Left\y.\Coercion{\pi_1}\x\y}{\Right\y.\Coercion{\pi_2}\x\y}
    \and
    \Coercion{
      \inferrule{
        \pi_1 :: \TypeA \subt \TypeA'
        \\
        \pi_2 :: \TypeB \subt \TypeB'
      }{
        \TypeA \tfork \TypeB
        \subt
        \TypeA' \tfork \TypeB'
      }
    }\x\y \triangleq
    \Receive\x\u.\Send\y\v{\Coercion{\pi_1}\u\v}{\Coercion{\pi_2}\x\y}  
    \quad\text{($u$ and $v$ fresh)}
    \and
    \Coercion{
      \inferrule{
        \pi :: \TypeA\subst{\sigma\X.\TypeA}\X \subt \TypeB
      }{
        \sigma\X.\TypeA \subt \TypeB
      }
    }\x\y \triangleq \Coercion\pi\x\y
    \and
    \Coercion{
      \inferrule{
        \pi :: A \subt B\subst{\sigma\X.B}\X
      }{
        A \subt \sigma\X.B
      }
    }\x\y \triangleq \Coercion\pi\x\y
  \end{mathpar}
  \caption{\label{tab:coercion} Coercion semantics of subtyping (selected equations).}
\end{table}

Following Pierce~\cite{Pierce02} we provide a \emph{coercion semantics} to our
subtyping relation by means of two translation functions, one on derivations of
subtyping relations $\TypeA \subt \TypeB$ and one on typing derivations
$\wtp{P}\Context$ that make use of subtyping.
The first translation is (partially) given in \Cref{tab:coercion}. The
translation takes a derivation $\pi$ of a subtyping relation $A \subt B$ --
which we denote by $\pi :: A \subt B$ -- and generates a process
$\Coercion\pi\x\y$ that transforms (the protocol described by) $A$ into (the
protocol described by) $B$. The translation is parametrized by the two channels
$x$ and $y$ on which the transformation takes place: the protocol $A$ is
``consumed'' from $x$ and reissued on $y$ as a protocol $B$. In
\Cref{tab:coercion} we show a fairly complete selection of cases, the remaining
ones being obvious variations.
It is easy to establish that $\wtp{\Coercion\pi\x\y}{x : \dual{A}, y : B}$ if $A
\subt B$. In particular, consider an infinite branch $\gamma \eqdef
(\wtp{\Coercion{\pi_i}{x_i}{y_i}}{x_i : \dual{A_i}, y : B_i})_{i\in\NatSet}$ in
the typing derivation of the coercion where $A_0 = A$ and $B_0 = B$.
This branch corresponds to an infinite branch $(A_i \subt B_i)_{i\in\NatSet}$ in
$\pi :: A \subt B$.
According to \Cref{def:subt}, either clause (1) or clause (2) holds for this
branch. Suppose, without loss of generality, that clause (1) holds. Then
$\min\set{ C \mid \existinf i\in\NatSet: A_i = C }$ is a $\mu$-type.
According to \Cref{tab:coercion} we have that $(x_i)_{i\in\NatSet}$ is a
$\nu$-thread of $\gamma$, hence $\gamma$ is a valid branch.
Note that in general $\Coercion\pi\x\y$ is (the invocation of) a recursive
process.

Concerning the translation of typing derivations, it is defined by the equation
\begin{equation}
  \label{eq:ttd}
  \sem{
    \inferrule{
      \pi_1 :: \wtp{P}{\Context, x : A}
      \\
      \pi_2 :: \wtp{Q}{\Context, x : B}
    }{
      \wtp{\Cut\x{P}{Q}}\Context
    }
  }
  =
  \inferrule{
    \inferrule*{
      \sem{\pi_1\subst\y\x}
      \\
      \wtp{
        \Coercion\pi\y\x
      }{
        y : \dual{A},
        x : B
      }
    }{
      \wtp{
        \Cut\y{P\subst\y\x}{\Coercion\pi\y\x}
      }{
        \Context, x : B
      }
    }
    \\
    \sem{\pi_2}
  }{
    \wtp{
      \Cut\x{
        \Cut\y{
          P\subst\y\x
        }{
          \Coercion\pi\y\x
        }
      }{
        Q
      }
    }{
      \Context
    }
  }
\end{equation}
where $\pi :: A \subt \dual{B}$ and extended homomorphically to all the other
typing rules in \Cref{tab:typing-rules}. Note that \eqref{eq:ttd} turns every
application of the \refrule\SubRule into two applications of the standard
\muMALL cut rule.
The validity of the resulting typing derivation follows immediately from that of
the original typing derivation and that for the coercion, as argued earlier.


Thanks to the correspondence between \Calculus's typing rules and \muMALL,
well-typed \Calculus processes are well behaved. In particular, processes that
are well typed in a singleton context are deadlock free.

\begin{theorem}[deadlock freedom]
  \label{thm:df}
  If $\wtp{P}{x : A}$ then $P$ is deadlock free.
\end{theorem}

Moreover, the cut elimination property of
\muMALL~\cite{BaeldeDoumaneSaurin16,Doumane17} can be used to prove that
well-typed \Calculus processes terminate, similarly to related
systems~\cite{LindleyMorris16,Derakhshan2022}.

\begin{theorem}[termination]
  \label{thm:termination}
  If $\wtp{P}\Context$ then $P$ is terminating.
\end{theorem}
\begin{proof}[Proof sketch]
  The typing derivation for $\wtp{P}\Context$ with the subtype coercion made
  explicit maps directly to a valid \muMALL proof. Every reduction step of $P$
  maps directly to one or more principal reductions in the \muMALL proof. The
  reason why we could have more than one principal reduction for each process
  reduction comes from our choice of not having an explicit process form
  triggering the unfolding of a fixed point (see \refrule{\FoldRule\sigma}).
  Now, suppose that $P$ has an infinite run. Then there would be an infinite
  sequence of reduction steps starting from $P$, hence an infinite sequence of
  cut reductions in the corresponding \muMALL proof, which contradicts
  \cite[Proposition 3.5]{Doumane17}. Thus every run of $P$ must be finite.
\end{proof}

Note that \Cref{thm:termination} only assures that a well-typed process will not
reduce forever, not necessarily that the final configuration of the process is
free of restricted sessions. These may occur guarded by a prefix concerning some
free channel in the process. We can formulate a property of ``successful
termination'' by combining \Cref{thm:df,thm:termination}.

\begin{corollary}
  If $\wtp{P}{x : \One}$ then $P$ eventually reduces to $\Close\x$.
\end{corollary}

We conclude this section with an example showing that the additional clauses of
\Cref{def:subt} are key to making sure that $\subt$ is a termination-preserving
subtyping relation.

\newcommand{\Chatter}{\mathsf{Chatter}}

\begin{example}
  \label{ex:chatter}
  Consider a degenerate client $\Let\Chatter\x{\Right\x.\Call\Chatter\x}$ that
  engages into an infinite interaction with $\Server$ from \Cref{ex:math} and
  let $C \eqdef \nu\X.(\One \choice X)$. The derivation
  \[
    \begin{prooftree}
      \[
        \[
          \[
            \mathstrut\smash\vdots
            \justifies
            \wtp{
              \Call\Chatter\x
            }{
              x : C
            }
            \using\refrule{\FoldRule\nu}
          \]
          \justifies
          \wtp{
            \Right\x.\Call\Chatter\x
          }{
            x : \One \choice C
          }
          \using\refrule\SelectRule
        \]
        \justifies
        \wtp{
          \Call\Chatter\x
        }{
          x : \One \choice C
        }
        \using\refrule\CallRule
      \]
      \justifies
      \wtp{
        \Call\Chatter\x
      }{
        x : C
      }
      \using\refrule{\FoldRule\nu}
    \end{prooftree}
  \]
  is valid since the only infinite branch contains a $\nu$-thread $(x,x,\dots)$
  along which we find infinitely many applications of \refrule{\FoldRule\nu}.
  If we allowed the relation $C \subt \dual{B}$ (\cf the discussion leading to
  \Cref{def:subt}) the composition $\Cut\x{\Call\Chatter\x}{\Call\Server{x,z}}$
  would be well typed and it would no longer be the case that well-typed
  processes terminate, as the interaction between $\Chatter$ and $\Server$ goes
  on forever.
  \eoe
\end{example}


\section{Concluding remarks}
\label{sec:conclusion}

We have defined a subtyping relation for session types as the precongruence that
is insensitive to the (un)folding of recursive types and such that $\Zero$ and
$\Top$ act as least and greatest elements.
Despite the minimalistic look of the relation and the apparent rigidity in the
syntax of types, in which the arity of internal and external choices is fixed,
$\subt$ captures the usual co/contra variance of labels thanks to the
interpretation given to $\Zero$ and $\Top$. Other refinement relations for
session types with least and greatest elements have been studied in the
past~\cite{Padovani10,Padovani16}, although without an explicit correspondance with logic.

Unlike subtyping relations for session
types~\cite{GayHole05,CastagnaDezaniGiachinoPadovani09,MostrousYoshida15,GhilezanEtAl21}
that only preserve \emph{safety properties} of sessions (communication safety,
protocol fidelity and deadlock freedom), $\subt$ also preserves termination,
which is a \emph{liveness property}. For this reason, $\subt$ is somewhat
related to \emph{fair subtyping}~\cite{Padovani13,Padovani16}, which preserves
\emph{fair termination}~\cite{GrumbergFrancezKatz84,Francez86}. It appears that
$\subt$ is coarser than fair subtyping, although the exact relationship between
the two relations is difficult to characterize because of the fundamentally
different ways in which recursive behaviors are represented in the syntax of
types. The subtyping relation defined in this paper inherits least and greatest
fixed points from \muMALL~\cite{BaeldeDoumaneSaurin16,Doumane17}, whereas fair
subtyping has been studied on session type languages that either make use of
general recursion~\cite{Padovani13} or that use regular trees
directly~\cite{Padovani16}. A more conclusive comparison is left for future
work.

A key difference between the treatment of fixed points in this work and a
related logical approach to session subtyping~\cite{Horne2020} is that, while
both guarantee deadlock freedom, the current approach also guarantees
termination. Insight concerning the design of fixed points should be exportable
to other session calculi independently from any logical interpretation. In
particular, it would be interesting to study subtyping for \emph{asynchronous
session types}~\cite{MostrousYoshida15,GhilezanEtAl21} in light of
Definition~\ref{def:subt}. This can be done by adopting a suitable coercion
semantics to enable buffering of messages as in simple
orchestrators~\cite{Padovani10b}.


\paragraph{Acknowledgments.}
We are grateful to the anonymous reviewers for their thoughtful comments.

\bibliographystyle{eptcs}
\bibliography{main}


\end{document}